\documentclass[letterpaper, 10 pt, conference]{ieeeconf}
\usepackage{cite}
\usepackage{amsthm}
\usepackage{amssymb}
\usepackage{amsfonts}
\usepackage{mathtools}
\usepackage{algorithm}
\usepackage{algpseudocode} 
\usepackage{graphicx}
\usepackage{textcomp}
\usepackage{threeparttable}
\usepackage{hyperref}
\usepackage{amssymb}
\usepackage{enumitem}
\usepackage{tabu}
\usepackage[utf8]{inputenc}
\usepackage{setspace}
\usepackage[font=footnotesize,labelfont=bf]{caption}
\usepackage[font=footnotesize,labelfont=bf]{subcaption}
\usepackage{multirow}
\usepackage{tabularx}
\usepackage{array}
\usepackage{makecell}
\usepackage{dblfloatfix}
\usepackage{tikz-network}
\tikzset{
block/.style = {draw, fill=white, rectangle, minimum height=2em, minimum width=5em},
tmp/.style  = {coordinate}, 
sum/.style= {draw, fill=white, circle, node distance=1cm},
input/.style = {coordinate},
output/.style= {coordinate},
pinstyle/.style = {pin edge={to-,thin,black}
}
}

\newtheorem{theorem}{\textbf{Theorem}}
\newtheorem{lemma}{\textbf{Lemma}}

\newtheorem{corollary}{\textbf{Corollary}}
\newtheorem{remark}{\textbf{Remark}}
\newtheorem{assumption}{\textbf{Assumption}}
\usepackage{authblk}
\usepackage{float}
\usepackage[table]{xcolor}
\newtheorem{definition}{\textbf{Definition}}
\usepackage{color}
\usepackage{tikz}
\IEEEoverridecommandlockouts 

\title{Local Stability and Region of Attraction Analysis for Neural Network Feedback Systems under Positivity Constraints}

\author{Hamidreza Montazeri Hedesh$^1$, Moh Kamalul Wafi$^1$, and Milad Siami$^1$\thanks{$^1$H. Montazeri Hedesh, M. K. Wafi, M. Siami are with the Department of Electrical \& Computer Engineering, Northeastern University, Boston, MA 02115, USA.
(e-mails: {\tt\footnotesize	\{montazerihedesh.h, wafi.m, m.siami\}@northeastern.edu}).}

 \thanks{This material is based upon work supported in part by the U.S. Office of Naval Research under Grant Award N00014-21-1-2431, in part by the U.S. National Science Foundation under Grant Award 2121121 and Grant Award 2208182, and in part by the DEVCOM Analysis Center and was accomplished under Contract Number W911QX-23-D0002. The views and conclusions contained in this document are those of the authors and should not be interpreted as representing the official policies, either expressed or implied, of the DEVCOM Analysis Center or the U.S. Government. The U.S. Government is authorized to reproduce and distribute reprints for Government purposes notwithstanding any copyright notation herein.}
}
\begin{document}
\maketitle

\begin{abstract}
We study the local stability of nonlinear systems in the Lur’e form with static nonlinear feedback realized by feedforward neural networks (FFNNs). By leveraging positivity system constraints, we employ a localized variant of the Aizerman conjecture, which provides sufficient conditions for exponential stability of trajectories confined to a compact set. Using this foundation, we develop two distinct methods for estimating the Region of Attraction (ROA): (i) a less conservative Lyapunov-based approach that constructs invariant sublevel sets of a quadratic function satisfying a linear matrix inequality (LMI), and (ii) a novel technique for computing tight local sector bounds for FFNNs via layer-wise propagation of linear relaxations. These bounds are integrated into the localized Aizerman framework to certify local exponential stability. Numerical results demonstrate substantial improvements over existing integral quadratic constraint-based approaches in both ROA size and scalability.
\end{abstract}
\allowdisplaybreaks
\section{Introduction}

Neural networks (NNs) are increasingly being used in control systems due to their expressive power and ability to learn complex nonlinear mappings\cite{HORNIK1989359}. There are recent interesting usage of NN structures in Control \cite{goel2024can,honarpisheh2025generalization}. However, certifying the stability of NN-in-the-loop systems remains a major challenge, particularly when such systems are deployed in safety-critical applications. Traditional Lyapunov-based analysis is often infeasible for high-dimensional or opaque NN controllers, motivating the need for scalable, structure-exploiting methods for stability verification.

Existing techniques that attempt to verify stability for NN feedback---such as those using integral quadratic constraints (IQCs)\cite{yin2021stability,de2023event,noori2024stability,noori2024complete}, sum-of-squares optimization \cite{newton2022stability}, passivity-based approaches \cite{jin2020stability}, circle criterion \cite{richardson2024strengthened} or global Lipschitz bounds \cite{szegedy2013intriguing}---either scale poorly with system size or result in overly conservative guarantees.
Recent advances such as RecurJac \cite{zhang2019recurjac}, CROWN \cite{wang2021beta}, and IBP \cite{gowal2018effectiveness} provide tools for bounding NN outputs, but these have not been explicitly integrated into closed loop system frameworks. There are other useful control methods such as density functions \cite{11012717} or positive system methods \cite{bill2016stability} that has not yet extended into verification of NN in the loop.

\emph{Positive systems} offer a favorable structure for a streamlined and scalable stability analysis. These systems have been widely explored in the literature of linear systems \cite{shafai2024positive,10708136}. However, their benefits in nonlinear systems have not been extensively studied. The positive Lur’e framework provides a well-established setting for studying absolute stability of complex nonlinear systems. For example, the positive Aizerman conjecture gives sufficient conditions for global exponential stability in this context when the feedback nonlinearity lies within a sector \cite{drummond2022aizerman}.
However, there is currently a lack of techniques that can leverage the structural properties of positive systems to enable \emph{scalable} and \emph{less conservative} verification of NN feedback stability.

This paper addresses this gap by proposing two tractable and scalable methods for local stability analysis and estimating the ROA of positive Lur’e systems with NN feedback, both grounded in a localized version of the positive Aizerman conjecture. The first method uses a Lyapunov-based approach with linear matrix inequalities (LMIs), applicable to general sector-bounded nonlinearities. The second introduces a novel, \emph{tight local sector bound for FFNNs}, which to the best of our knowledge, is the first such formulation in the literature. This bound is propagated layer by layer through the network and naturally fits within the Aizerman framework, enabling a scalable and certifiable ROA analysis for NN feedback systems. Compared to existing bounds such as RecurJac, CROWN, and IBP, this novel local sector bound—built upon our previously proposed global sector bound \cite{Hamidrezaglobalsectorbound} and its application in robust stability of NN feedback loops \cite{hedesh2025robust}—features a nonaffine linear shape, and is specifically developed for scalable local stability analysis of NN feedback loops.
The key contributions of this paper are:
    (i) A localized variant of the Positive Aizerman Conjecture tailored for local stability analysis and further calculating the corresponding ROA;
    (ii) A Lyapunov-based method for local stability analysis and ROA estimation in positive Lur’e systems with general sector-bounded feedback;
    (iii) A novel local sector bound formulation for FFNNs, and further integration of it in a local stability analysis framework for NN feedback loops.

These contributions are significant since they provide a formal guarantee that despite the nonlinear nature of the NN, there is a specific area in the state space where the system is stabilizable. This result highlights the practical utility of NN controllers in real-world applications, where local stability is often sufficient. Through rigorous analysis and simulation, this paper provides both theoretical insights and practical guidelines for implementing NN controllers in feedback systems. This opens up new possibilities for using machine learning techniques in classical control theory, where traditional linear control methods fall short \cite{10412393,10903652}.

Finally, numerical results demonstrate that our methods outperform existing IQC-based approaches \cite{yin2021stability} in ROA size and computational efficiency, making them suitable for larger systems where existing verification methods are impractical.

\subsection{Notations}
We denote the set of real numbers, real vectors of dimension $n$, and real matrices of dimension $n\times m$ by $\mathbb{R}$, $\mathbb{R}^n$, and $\mathbb{R}^{n\times m}$. The orders $>,<,\geq$, and $\leq$ are interpreted element-wise when applied to vectors and matrices. In addition, the operators $\succ$ and $\prec$ denote positive and negative definite matrices. Subsequently, a doubly positive matrix $A$ is both $A\succ 0$ and $A>0$. The symbol $\mathbf{0}$ denotes a zero vector of appropriate dimension. We also define $\mathbb{R}_+^{n\times m} := \{A\in\mathbb{R}^{n\times m}\mid A\geq 0\}$. The sets $\mathbb{R}_+^{n}$ and $\mathbb{R}_+$ are defined similarly. A Metzler matrix is characterized by having nonnegative off-diagonal elements.

\section{Proposed Theoretical Framework}
First, we lay the theoretical foundations for presenting our results, encompassing positive Lur’e systems and the positive Aizerman conjecture. throughout the paper, we briefly use positivity for internal positivity. The following definition and lemmas along with their proof can be found in \cite{farina2011positive}.

\subsection{Positive Lur’e Systems}
Consider a linear time invariant (LTI) system of the form:
\begin{equation}\label{eq:generallti}
    \dot x(t) = A x(t) + B u(t), \quad
    y(t) = C x(t),
\end{equation}
where $A \in \mathbb R^{n\times n}$, $B \in \mathbb R^{n\times m}$, and $C \in \mathbb R^{p\times n}$ are constant matrices. The vectors $x(t) \in \mathbb R^n$, $u(t) \in \mathbb R^m$, and $y(t) \in \mathbb R^p$ denote the state, input, and output variables, respectively,  with initial condition $x(0)$.
\begin{definition}\label{def:positive system}
    The system in \eqref{eq:generallti}, is called ``positive'' if, \( \forall t > 0 \), we have \( x(t) \geq 0 \), given \( x(0)\coloneqq x_0 \geq 0 \) and \( u(t) \geq 0 \).
\end{definition}
\begin{lemma}
    The system in \eqref{eq:generallti} is positive if and only if \( A \) is a Metzler matrix, $B \in \mathbb R_+^{n\times m}$, and $C \in \mathbb R_+^{p\times n}$.
\end{lemma}
\begin{lemma}\label{lem2}
    If a positive system in~\eqref{eq:generallti} is also asymptotically stable, then there exists a vector \( v \in \mathbb{R}_+^n \), with \( v > 0 \), such that \( v^\top A < 0 \).
\end{lemma}

We refer the reader to \cite{drummond2022aizerman} and the references therein for the proof of Lemma~\ref{lem2}. For notational simplicity, we omit the explicit time dependence and write \( x \), \( y \) and \( u \) in place of \( x(t) \), \( y(t) \), and \( u(t) \) whenever the meaning is clear from context.

To define a positive Lur’e system, we consider the system in \eqref{eq:generallti} under a static nonlinear feedback represented as $u = \Phi(Cx,.)$, as illustrated in Fig. \ref{fig:luresystem}. The resulting closed-loop dynamics can be written as: 
\begin{equation}\label{eq:positiveluresystem}
    \dot{x} = Ax + B \Phi(Cx, t),
\end{equation}
where \( \Phi: \mathbb{R}^p \times \mathbb{R} \to \mathbb{R}^m \) is a multivariate function that satisfies the following assumption. The assumption ensures the existence and uniqueness of solutions.
\begin{assumption}\label{assu1}
    There exists a unique, locally absolutely continuous function $\chi: \mathbb{R}_+ \to \mathbb{R}^n$ that satisfies \eqref{eq:positiveluresystem} almost everywhere for any $x_0 \in \mathbb{R}^n$. Moreover, $\Phi(0, t) = 0, \forall t \geq 0$, ensuring that $x_* = 0$ is an equilibrium point of \eqref{eq:positiveluresystem}. The function $\Phi(z, t)$ is locally Lipschitz in $z$ and measurable in $t$, with certain mild boundedness assumptions \cite[Thm. 54, Prop. C.3.8]{sontag2013mathematical}.
\end{assumption}

Furthermore, to enable stability analysis, it is necessary to characterize the sector bounds of the nonlinear component. For a multi-input multi-output (MIMO) positive system, these bounds are best described using component-wise inequalities. The function $\Phi$ is said to be sector bounded in $[\Sigma_1, \Sigma_2]$ if:
\begin{equation}\label{eq:mimo sector bound}
    \Sigma_1 z \leq \Phi(z, t) \leq \Sigma_2 z, \quad \forall z \in \mathbb{R}^p_+, \forall t \geq 0,
\end{equation}
where $\Sigma_1, \Sigma_2 \in \mathbb{R}^{m \times p}$, and $\Sigma_1 \leq \Sigma_2$.

A positive Lur'e system is defined as a system of the form \eqref{eq:positiveluresystem}, where the nonlinearity $\Phi(\cdot)$ satisfies the sector condition \eqref{eq:mimo sector bound}, and the overall system meets the criteria of Definition \ref{def:positive system}. With this foundation, we are now ready to introduce the positive Aizerman conjecture.

\subsection{Positive Aizerman Conjecture}
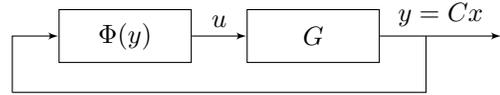
\begin{figure}[t!]
    \centering
    \begin{tikzpicture}[auto, node distance=2cm,>=latex']
        \node [tmp] (tmp1) {};
        \node [tmp, below of=tmp1,node distance=.75cm] (tmp) {};
        \node [block, right of=tmp1,node distance=1.5cm] (Control) {$\Phi(y)$};
        \node [block, right of=Control,node distance=2.5cm] (Plant) {$G$};
        \node [tmp, right of=Plant,node distance=2.5cm] (tmp2) {};
        \node [tmp, right of=Plant,node distance=1.5cm] (tmp3) {};
        \draw [->] (Control) -- node{$u$} (Plant);
        \draw [->] (Plant) -- node{$y=Cx$} (tmp2);
        \draw [-] (tmp3) |- (tmp);
        \draw [->] (tmp) |- (Control);
    \end{tikzpicture}
    \caption{\small Lur'e system with plant $G$ and nonlinear controller $\Phi$.}\vspace{-.4cm}
    \label{fig:luresystem}
\end{figure}
The Aizerman conjecture is a classical tool for absolute stability analysis in Lur’e systems. Although counterexamples exist for the general case, it holds under more stringent conditions—particularly positive systems \cite{bill2016stability}.

The following theorem, originally presented and proved in \cite{drummond2022aizerman}, formalizes the validity of Aizerman’s conjecture for MIMO positive systems.
\begin{theorem}\label{the:positiveaizerman}
    Consider a Lur'e system in \eqref{eq:positiveluresystem}, where $B, C \geq 0,$ and $\Phi$ is sector bounded in the interval $[\Sigma_1$, $\Sigma_2]$ in the sense of \eqref{eq:mimo sector bound}. The system is globally exponentially stable if $A+B\Sigma_1C$ is Metzler and $A+B\Sigma_2C$ is Hurwitz.
\end{theorem}

With the preliminaries in place, we now proceed to formally define the problem and outline our approach for analyzing the local stability and ROA of nonlinear Lur’e systems.

\subsection{Problem Formulation}
Consider the Lur'e system as defined in \eqref{eq:positiveluresystem}, with $B, C \geq 0,$ and the equilibrium point $x_*$ satisfying $\mathbf{0} = Ax_* + Bu_*$, $y_* = Cx_*$, $ u_* = \Phi(y_*)$.
Let $\chi(t; x_0)$ denote the solution to the closed-loop system at time $t$, initialized at $x_0$. Our goal is to analyze the stability of the equilibrium point $x_*$, and to identify a maximal underapproximation of the set of initial conditions from which the system trajectories converge to this point. This set is commonly referred to as the ROA in nonlinear systems literature and is defined as:
    \(\mathcal{R} := \{ x_0 \in \mathbb{R}^{n} : \lim_{t \to \infty} \chi(t; x_0) = x_* \}.\)

In the following sections, we first formulate the localized positive Aizerman conjecture and derive an associated inner approximation of the ROA. Building on this foundation, we then propose two complementary methods to compute inner approximations of ROA for positive Lur’e systems with FFNNs in the feedback loop: (i) a Lyapunov-based approach leveraging locally sector bounded nonlinearity assumptions and LMI optimization, and (ii) a novel NN sector bound, integrated into the framework of the local positive Aizerman conjecture to ensure local stability and estimate the ROA.


\section{Local Stability via Positive Aizerman Conjecture}
\vspace{-.1cm}We begin by introducing a localized sector condition for nonlinearities, which plays a central role in our analysis.
\begin{definition}[$\Gamma-$Sector Bounded Function]\label{def:gammasectorbound}
    Given an interval $[\Sigma_1, \Sigma_2]$ with $\Sigma_1, \Sigma_2 \in \mathbb{R}^{m\times p}$ and $\Sigma_1 \le \Sigma_2$, a multivariate nonlinear function $\Phi(y, \cdot): \mathbb{R}^p \times \mathbb{R} \to \mathbb{R}^m $ is said to be $\Gamma-$sector bounded within that interval if, for a compact and connected set \( \Gamma \subseteq \mathbb{R}^p \), the following condition holds:
\begin{equation}
    \Sigma_1 y \leq \Phi(y, .) \leq \Sigma_2 y, \quad \forall y \in \Gamma.
\end{equation}
\end{definition}
 Using this definition, the following theorem establishes exponential stability of output trajectories that remain inside the defined $\Gamma$ set.

\begin{theorem}\label{the:mathmanip}
    Consider the Lur’e system given in \eqref{eq:positiveluresystem}, where $B, C \geq 0$ and the nonlinearity $\Phi$ is $\Gamma-$sector bounded within the interval $[\Sigma_1,\Sigma_2]$ with $y_* \in \Gamma$. Suppose that $A+B\Sigma_1C$ is Metzler and $A+B\Sigma_2C$ is Hurwitz; then all the output trajectories that remain within the set $\Gamma$ are exponentially stable.
\end{theorem}
For the sake of brevity, we omit the full proof here; it follows a similar reasoning to the global positive Aizerman conjecture as presented in \cite{drummond2022aizerman}, with the additional assumption that the system trajectory remains confined within the sector bounded region \( \Gamma \). Following Theorem~\ref{the:mathmanip}, Lemma~\ref{lem:mathmanip} provides a sufficient condition for identifying a subset of initial states that guarantees the output trajectory remains within $\Gamma$ set.

\begin{lemma}\label{lem:mathmanip}
    Consider a multi-input single-output (MISO) Lur'e system satisfying theorem \ref{the:mathmanip}. Let $\overline y$ denote the upper bound
    of the set $\Gamma$ in which the nonlinearity is sector bounded. Then, an underapproximation of the ROA is given by the set of initial conditions $x_0$ satisfying:
    \begin{equation}\label{eq:ROA}
        Cx_0 \leq \frac{v_m}{v_M} \,\overline y,
    \end{equation}
    where $v>0$ satisfies $v^\top(A+B\Sigma_2C)<0$, and $v_m$ and $v_M$ are the smallest and largest element of $v$, in turn.
\end{lemma}

\begin{proof}
We can upper bound the system dynamics in \eqref{eq:positiveluresystem} as:
\begin{equation}\label{eq:midcalcul1}
    \dot{x} = (A + B\Sigma_2C)x + B(\Phi(Cx,t) - \Sigma_2Cx) \leq Mx,
\end{equation}
where \( M := A + B \Sigma_2 C \) is Metzler and Hurwitz. By Lemma \ref{lem2}, $\exists v>0$ such that $v^\top M < 0$ and considering the Lyapunov candidate $V(x) = v^\top x$, then its time derivative along system trajectories is:
\begin{equation}\label{eq:midcal2iuse}
        \dot{V}(x) = v^\top \dot{x} \le v^\top Mx.
\end{equation}
Since \( v^\top M < 0 \) and \( x \geq 0 \), we conclude that $\dot{V}(x) \leq 0 $, implying \( v^\top x \) is non-increasing over time. This gives:
\begin{equation}\label{eq:midcalcul3}
    v_m x \leq v^\top x \leq v^\top x_0 \leq v_M x_0.
\end{equation}
Multiplying both sides of $v_m x \le v_M x_0$ by $C$ yields $v_m y \le v_M y_0$. Now, if $v_M y_0 \le v_m \overline y$, then $v_m y \le v_m \overline y$, implying $y \le \overline y$. Therefore, all output trajectories starting from 
\begin{equation}
    y_0 \le \frac{v_m}{v_M}\,\overline y
\end{equation}
will remain below the threshold $\overline y$ and remain inside $\Gamma$, thus ensuring asymptotic stability.
\end{proof}
\begin{remark}
    The Lyapunov function presented in the proof of Lemma~\ref{lem:mathmanip} with equations \eqref{eq:midcalcul1} and \eqref{eq:midcal2iuse}, alongside the result of~\cite[Lem. 2]{drummond2022aizerman}, provides a proof for our Theorem~\ref{the:mathmanip}.
\end{remark}


\section{Lyapunov Based Method}
An alternative approach to Lemma \ref{lem:mathmanip} involves the construction of Lyapunov functions benefitting from sector bounded nonlinearity constraints. The following result establishes a verified quadratic Lyapunov function for general sector bounded systems.

\begin{theorem}[Quadratic certificate]\label{the:lyapaizerman}
For the Lur’e system \eqref{eq:positiveluresystem} with \(B,C\ge 0\) and a nonlinearity \(\Phi\) that is \(\Gamma\)-sector bounded on \([\Sigma_1,\Sigma_2]\),
if \(A+B\Sigma_1C\) is Metzler and \(A+B\Sigma_2C\) is Hurwitz, then there exists a \emph{doubly positive}
\(P\in\mathbb{R}^{n\times n}\) such that
\begin{align}
    (A + B\Sigma_2 C)^\top P + P(A + B\Sigma_2 C) \prec 0. \label{eq:condit3}
\end{align}
Consequently, \(V(x)=x^\top P x\) is a Lyapunov function certifying local stability; moreover, for some \(\alpha>0\),
the sublevel set \(\mathcal{V}_\alpha=\{x: V(x)\le\alpha\}\) is compact, positively invariant, and forms an inner
approximation of the ROA.
\end{theorem}


Moreover, the sublevel sets of \( V \) define compact, positively invariant regions around the origin that constitute an inner approximation of the ROA for the nonlinear system. That is, there exists \( \alpha > 0 \) such that the set \( \mathcal{V}_\alpha := \{ x \in \mathbb{R}^n \mid V(x) \leq \alpha \} \) is entirely contained within the ROA.

\begin{proof}
Consider the quadratic Lyapunov candidate \( V(x) = x^\top P x \), where \( P \) is a doubly positive matrix satisfying the LMI in \eqref{eq:condit3}.
Define the time derivative of \( V \) along system trajectories:
\begin{equation}\label{eq:vdotoflyap}
    \dot{V}(x) = x^\top (A^\top P + P A) x + 2 x^\top P B \Phi(Cx,.).
\end{equation}

By the assumption that the nonlinearity \( \Phi \) is \(\Gamma-\)sector bounded in the interval \([\Sigma_1, \Sigma_2]\), we obtain the inequality:
\begin{equation}\label{eq:ineq_phi}
    \Phi(Cx,.) \leq \Sigma_2 Cx, \quad \forall x \in \mathcal{D},
\end{equation}
for some neighborhood \(\mathcal{D}\) of the origin. Noting that \( B \), \( P \), and \( x \) are all nonnegative element-wise and using the monotonicity of scalar multiplication under the element-wise inequality \eqref{eq:ineq_phi}, we deduce:
\begin{equation}\small
    2 x^\top P B \Phi(Cx,.) \leq 2 x^\top P B \Sigma_2 C x.
\end{equation}

Substituting this into the expression for \( \dot{V}(x) \), we obtain:
\begin{align}\small
    \dot{V}(x) &\leq x^\top (A^\top P + P A) x + 2 x^\top P B \Sigma_2 C x \nonumber \\
    &= x^\top \left[ (A + B \Sigma_2 C)^\top P + P (A + B \Sigma_2 C) \right] x.
\end{align}

By construction of \( P \), the matrix inequality \eqref{eq:condit3} implies that the right-hand side is strictly negative. Hence, \( \dot{V}(x) < 0 \) for all \( x \neq 0 \) in some neighborhood of the origin, completing the proof of local asymptotic stability.
\end{proof}

Building upon the stability guarantee provided in Theorem~\ref{the:lyapaizerman}, we now proceed to define an inner approximation of the ROA for the nonlinear system. It is important to note that, in general, one cannot assert that the entire sector bounded domain \( \Gamma \) corresponds to the ROA. This is due to the possibility that certain trajectories originating within \( \Gamma \) may exit the sector before asymptotically approaching the equilibrium point, thereby violating the required sector condition \( \Sigma_1 y \leq \Phi(y,.) \leq \Sigma_2 y \) for all \( t \geq 0 \).

The identification of a subset of \( \Gamma \) that remains forward-invariant under the flow and for which the sector condition is preserved globally in time is, in general, a nontrivial task, as it necessitates precise characterization of invariant sets embedded within \( \Gamma \). To circumvent this difficulty, we leverage the candidate Lyapunov function proposed in Theorem~\ref{the:lyapaizerman}, and construct a Lyapunov-based estimate of the ROA via its sublevel sets.

Specifically, we define quadratic sublevel sets of the form
\(
    \mathcal{L}_\rho \coloneqq \left\{ x \in \mathbb{R}^n \,\middle|\, x^\top P x \leq \rho \right\},
\)
where \( P \) is a doubly positive matrix satisfying the LMI condition in \eqref{eq:condit3}. Instead of exhaustively verifying the sector condition for all trajectories in \( \Gamma \), we restrict our attention to determining whether the Lyapunov function decreases along system trajectories within a candidate sublevel set. In particular, we evaluate the condition
\begin{equation}\label{eq:vdotnn}\small
    \dot{V}(x) = x^\top (A^\top P + P A) x + 2 x^\top P B \Phi(Cx,.) < 0,
\end{equation}
on the boundary of \( \mathcal{L}_\rho \), i.e., for all \( x \in \partial \mathcal{L}_\rho \). If \eqref{eq:vdotnn} holds uniformly on the boundary, then the sublevel set \( \mathcal{L}_\rho \) is forward-invariant and contained within the ROA. Consequently, the largest value of \( \rho \) for which this condition is satisfied defines the maximal certified invariant level set, which we take as an estimate of the ROA. This procedure is formalized in Algorithm~\ref{alg:lyapunov_roa}.

\begin{remark}
While this work primarily utilizes a quadratic Lyapunov function of the form \( V(x) = x^\top P x \) to certify local stability and construct ROA estimates, the proposed methodology can also be applied using a linear copositive Lyapunov function \( V(x) = v^\top x \), where \( v > 0 \). In particular, under the same sector bounded conditions, if \( A + B\Sigma_1 C \) is Metzler and \( A + B\Sigma_2 C \) is Hurwitz, then there definitely exists a vector \( v \) satisfying the inequality
\(
(A + B\Sigma_2 C)^\top v < 0,
\)
which guarantees local asymptotic stability in the positive orthant. The existence of such $v$ is guaranteed by Lemma~\ref{lem2}. This linear formulation offers a computationally efficient alternative, albeit often yielding more conservative ROA estimates than the quadratic case.
\end{remark}
\begin{algorithm}[t!]
\caption{ROA Estimation via Lyapunov Sublevel Sets for Lur'e Systems}
\label{alg:lyapunov_roa}
{\small\begin{algorithmic}[1]
\Require Lur'e system matrices $A, B > 0, C > 0$, nonlinearity $\Phi(y,.)$ $\Gamma$-sector bounded in $[\Sigma_1, \Sigma_2]$
\Ensure Estimate of ROA
\State \textbf{Check Sector Conditions:}
\If{$A + B\Sigma_1C$ is Metzler and $A + B\Sigma_2C$ is Hurwitz}
    \State Proceed to Lyapunov analysis
\Else
    \State \textbf{Exit:} Sector conditions not satisfied
\EndIf
\State \textbf{Solve LMI:} Find $P \succ 0$, $P > 0$ satisfying
\[
(A + B\Sigma_2C)^\top P + P(A + B\Sigma_2C) \prec 0
\]
\If{no such $P$ exists}
    \State \textbf{Exit:} No suitable Lyapunov function found
\EndIf
\State \textbf{Define Lyapunov Function:} $V(x) = x^\top P x$
\State \textbf{Initialize} $\rho_{\text{max}} \gets 0$
\For{$\rho$ in increasing values}
    \State Define sublevel set $\mathcal{L}_\rho = \{x \mid V(x) \leq \rho\}$
    \State Sample points $x$ on the boundary $\partial \mathcal{L}_\rho$
    \For{each $x$ on $\partial \mathcal{L}_\rho$}
        \State Compute $\dot{V}(x)$ using 
        \[
        \dot{V}(x) = x^\top (A^\top P + P A) x + 2x^\top P B \Phi(Cx,.)
        \]
        \If{$\dot{V}(x) \geq 0$}
            \State \textbf{Break}: Sublevel set not invariant
        \EndIf
    \EndFor
    \If{$\dot{V}(x) < 0$ for all $x$ on $\partial \mathcal{L}_\rho$}
        \State Update $\rho_{\text{max}} \gets \rho$
    \Else
        \State \textbf{Break}: Maximum invariant sublevel set found
    \EndIf
\EndFor
\State \Return $\mathcal{L}_{\rho_{\text{max}}}$ as an estimate of ROA
\end{algorithmic}}
\end{algorithm}

This method is applicable to general nonlinearities, including NNs. It significantly reduces the search space for ROA from arbitrary sets to ellipsoidal regions aligned with Lyapunov levels. By focusing on invariant sublevel sets rather than the global sector region $\Gamma$, the approach mitigates the conservatism typically associated with worst-case sector analysis and provides a practically verifiable estimate of the region of attraction.

\section{Local Sector Bounds for FFNN}
Given that the positive Aizerman conjecture fundamentally relies on the existence of sector bounds for the system nonlinearity, we begin by introducing a novel formulation of local sector bounds for NNs specifically tailored to the stability analysis theorems. To the best of our knowledge, this is the first sector based characterization of its kind within the machine learning literature.

Relative to other available bounds—norm products \cite{szegedy2013intriguing}, CROWN \cite{wang2021beta}, RecurJac \cite{zhang2019recurjac}, and local Lipschitz estimates \cite{shi2022efficiently}—our construction (i) is tighter than norm-based bounds, (ii) avoids CROWN’s affine offsets, aligning with the linear sector condition in the Aizerman framework, and (iii) fixes sector slopes at the origin rather than taking derivative of the activation function at boundaries of pre-activation logits, cutting computational overhead while preserving fidelity near equilibrium. The method scales to high-dimensional networks without symbolic or gradient calculations. We focus on fully connected FFNNs without biases, although these constraints can be relaxed with a straightforward extension of the method. The network structure is defined as follows:




\vspace{-.3cm}
\begin{subequations}\label{eq:NNcontroller}\small
\begin{align}
    &\omega^{(0)} = y,\\
    &\nu^{(i)} = W^{(i)}\omega^{(i-1)},\quad \omega^{(i)} = \phi^{(i)}(\nu^{(i)}), \quad i=1,\dots,q \label{eq:qthlayer}\\
    &u = W^{(q+1)}\omega^{(q)},
\end{align}
\end{subequations}
where the input is denoted by \( y \in \mathbb{R}_+^{n_0} \), and the output of the \( i \)-th layer (pre-activation) is given by \( \nu^{(i)} \in \mathbb{R}^{n_i} \) for \( i = 1, \dots, q \). Moreover, \( W^{(i)} \in \mathbb{R}^{n_i \times n_{i-1}} \) is the weight matrix at layer \( i \), and \( \phi^{(i)}(\cdot) \) denotes the element-wise activation function applied at that layer.

To compute the local sector bounds, we employ \textit{interval-based nonaffine linear relaxations} of the activation functions, resulting in linear upper and lower approximations of their nonlinear mappings. These linear bounds are then \textit{forward propagated} through the weight matrices at each layer. This process maintains a linear relation between the input and the propagated bounds at every stage of the network.


We start from sector bounding the first layer. For the first layer, the pre-activation logits  $(\nu^{(1)} = W^{(1)}y)$ are trivially sector bounded in:
\begin{equation}\label{eq:bound1}\small
   W^{(1)} y \leq\nu^{(1)}\leq W^{(1)} y.
\end{equation}

Now we propagate this sector bound through the activation function. An explanation of this stage is given in the following subsection.

\subsection{Propagation of Sector Bounds through Activation Functions}\label{subsec:linearrelxation}
To carry out the linear relaxations for the activation functions, we first need to compute intervals for the pre-activation logits \( \nu^{(i)} \in \mathbb{R}^{n_i} \) at each hidden layer, denoted by \( \underline{\nu}^{(i)}, \overline{\nu}^{(i)} \in \mathbb{R}^{n_i} \). These intervals serve as the foundation for constructing valid linear relaxations of the activation functions. The full procedure for calculating these bounds across the network is summarized in part \ref{subsec:preactivation}, with additional technical details available in \cite{gowal2018effectiveness}. Given the pre-activation logits calculated from part \ref{subsec:preactivation}, we propagate sector bounds through the activation function in part \ref{subsubsec:tanhrelax}.

\subsubsection{Computation of pre-Activation logits}\label{subsec:preactivation}
Consider a fully connected FFNN described in \eqref{eq:NNcontroller}. The pre-activation bounds are computed recursively using interval arithmetic as follows.

Assuming the input \( y \in \Gamma \subseteq \mathbb{R}^{n_0} \) is bounded component-wise by \( \underline{y} \leq y \leq \overline{y} \), the input bounds propagate through the first layer using:
\begin{equation}
    \underline{\nu}^{(1)} = \min_{y \in [\underline{y}, \overline{y}]} W^{(1)} y, \quad
    \overline{\nu}^{(1)} = \max_{y \in [\underline{y}, \overline{y}]} W^{(1)} y.
\end{equation}
Since the weights are fixed, these expressions reduce to standard interval linear operations. Specifically, for each component of $\nu^{(1)}$:
\begin{equation}\small
    \underline{\nu}_j ^{(1)} = \sum_{k=1}^{n_0} \min\left\{ W^{(1)}_{j k} \underline{y}_k, W^{(1)}_{j k} \overline{y}_k \right\},\end{equation}
\begin{equation}\small
    \overline{\nu}_j ^{(1)} = \sum_{k=1}^{n_0} \max\left\{ W^{(1)}_{j k} \underline{y}_k, W^{(1)}_{j k} \overline{y}_k \right\}.
\end{equation}

For subsequent layers \( i \geq 2 \), we apply the same logic to the output bounds of the previous layer's post-activation:
\begin{align}\small
    \underline{z}^{(i-1)} &= \phi^{(i-1)}(\underline{\nu}^{(i-1)}), \quad 
    \overline{z}^{(i-1)} = \phi^{(i-1)}(\overline{\nu}^{(i-1)}), \\
    \underline{\nu}^{(i)} &= \min_{z \in [\underline{z}^{(i-1)}, \overline{z}^{(i-1)}]} W^{(i)} z, \\
    \overline{\nu}^{(i)} &= \max_{z \in [\underline{z}^{(i-1)}, \overline{z}^{(i-1)}]} W^{(i)} z.
\end{align}
These recursive bounds can be computed efficiently, and they enable us to define the domain over which linear relaxations of activation functions are constructed. The resulting pre-activation intervals \( [\underline{\nu}^{(i)}, \overline{\nu}^{(i)}] \) are then used to form upper and lower linear envelopes for each activation function.

\begin{figure*}[t]
    \centering
    \begin{subfigure}[t]{0.32\textwidth}
        \centering
        \includegraphics[width=\textwidth]{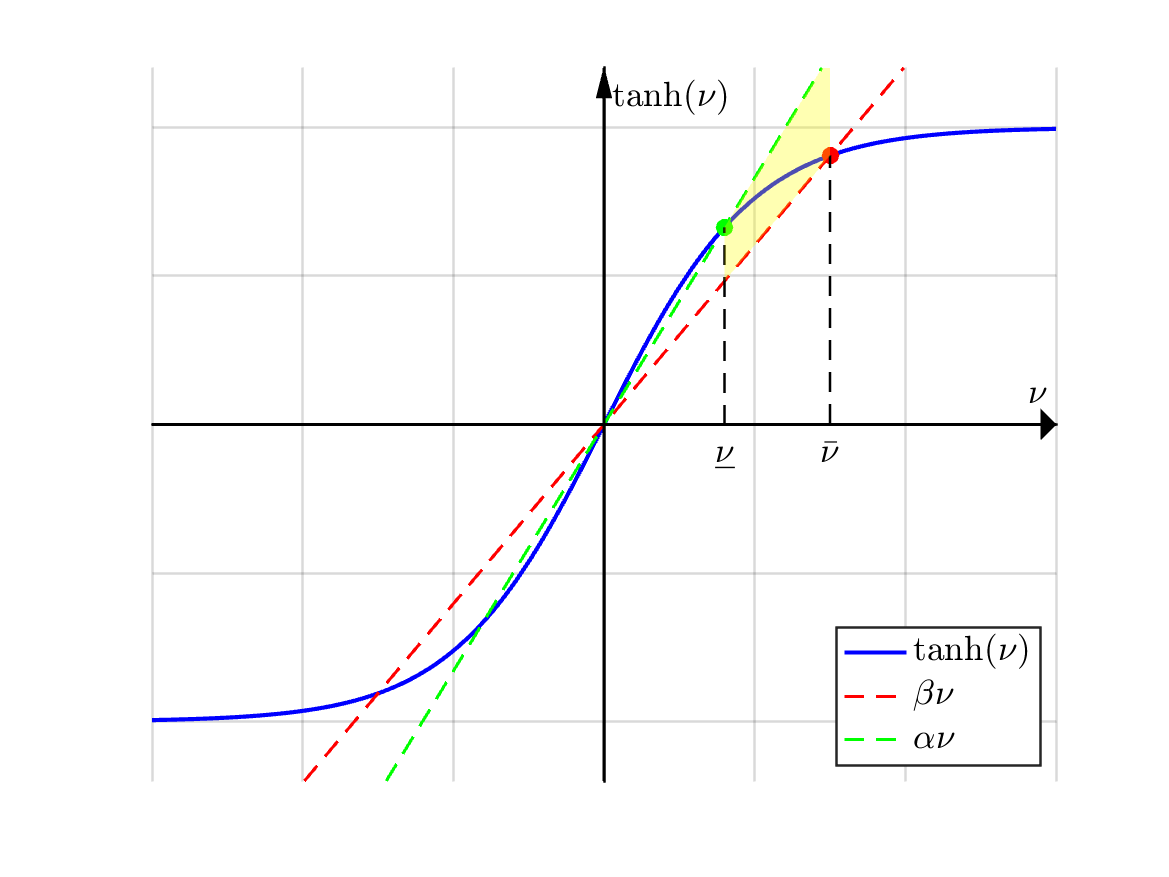}
        \caption{ Positive $\nu$: $\beta = \frac{\tanh(\overline{\nu})}{\overline{\nu}},\alpha = \frac{\tanh(\underline{\nu})}{\underline{\nu}},$ \\ \centering$\beta \underline{\nu} \le\tanh(\nu)\le \alpha \overline{\nu}$}
        \label{fig:posnu}
    \end{subfigure}
    \begin{subfigure}[t]{0.32\textwidth}
        \centering
        \includegraphics[width=\textwidth]{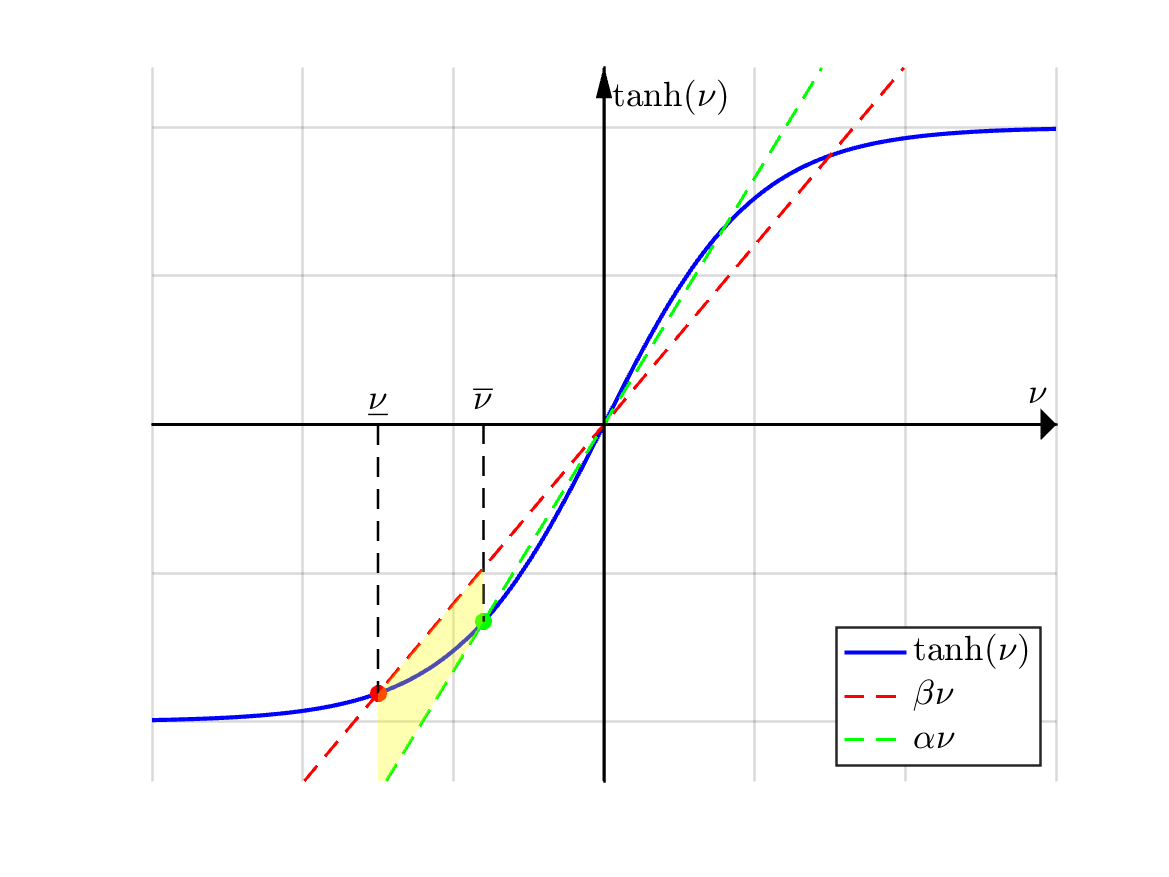}
        \caption{Negative $\nu$: $\beta = \frac{\tanh(\underline{\nu})}{\underline{\nu}},\alpha = \frac{\tanh(\overline{\nu})}{\overline{\nu}},$ \\ \centering$\alpha \underline{\nu} \le\tanh(\nu)\le \beta \overline{\nu}$}
        \label{fig:negnu}
    \end{subfigure}
    \begin{subfigure}[t]{0.32\textwidth}
        \centering
        \includegraphics[width=\textwidth]{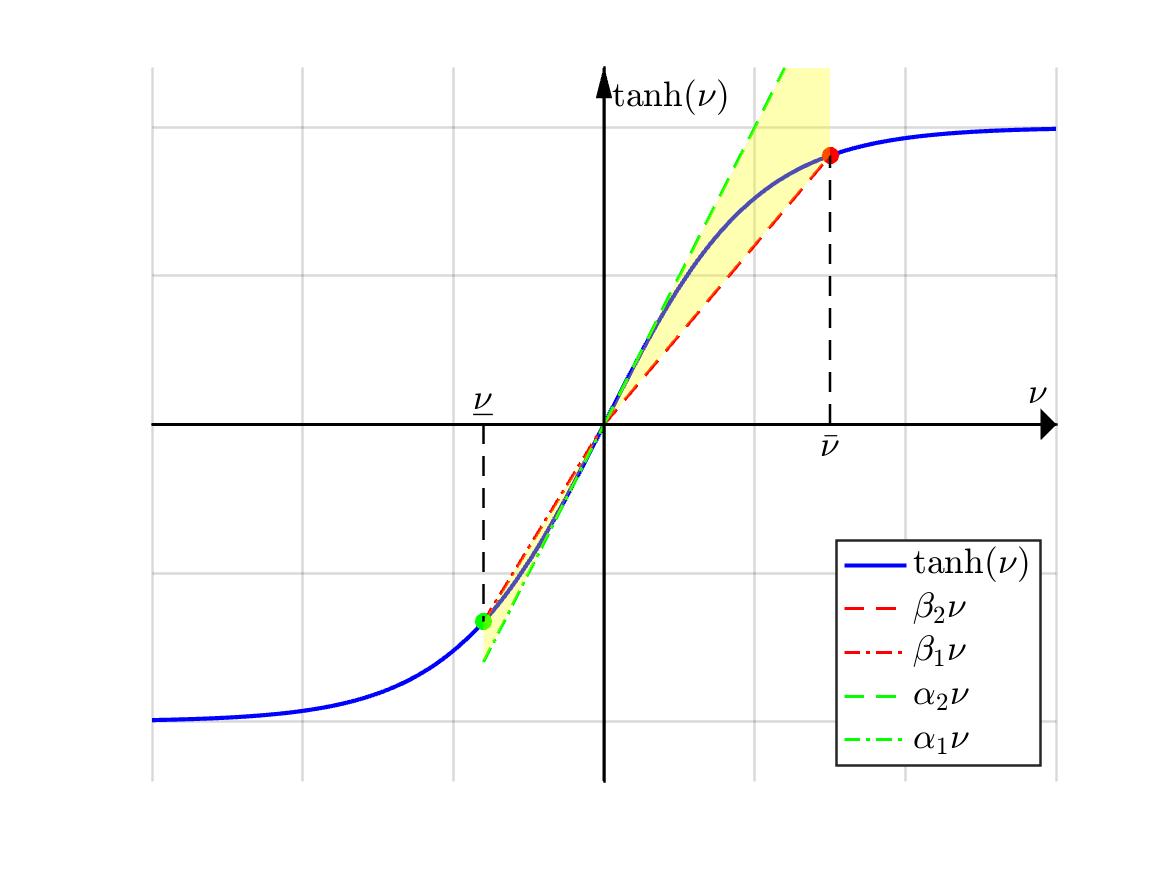}
        \caption{Unstable $\nu$: $\alpha_1 = \alpha_2 = \frac{d\tanh(\nu)}{d(\nu)} = 1,$\\ \centering$\beta_1 = \frac{\tanh(\underline{\nu})}{\underline{\nu}},\beta_2 = \frac{\tanh(\overline{\nu})}{\overline{\nu}},$ \\ \centering $\alpha_1\underline{\nu} \le \tanh{(\nu)} \le \alpha_2 \overline{\nu}$}
        \label{fig:unstablenode}
    \end{subfigure}
    \vspace{-.2cm}
    \caption{\small Illustration of linear relaxation of $\phi = \tanh(\nu)$ under different intervals for $\nu$.\vspace{-.6cm}}
    \label{fig:top_three_images}
\end{figure*}
With the pre-activation logit bounds computed, we are now equipped to propagate sector bounds through activation functions, as shown in the following part.
\subsubsection{Sector relaxation of activation functions}\label{subsubsec:tanhrelax}
To obtain a linear sector representation of the activation outputs \( \omega^{(i)} = \phi(\nu^{(i)}) \), we apply linear relaxation techniques to the nonlinear activation function.
Assume that the pre-activation logits are sector bounded in
\begin{equation}\label{eq:firstoftanh}
 L^{(i)} y \leq \nu^{(i)} \leq U^{(i)} y
 \end{equation}
We readily have this assumption met for the first layer as shown in \eqref{eq:bound1}.
We now derive the corresponding sector bounds for the post-activation output \( \omega^{(i)} \).

Consider a scalar neuron \( j \) at layer \( i \), where the activation output is given by \( \omega_j^{(i)} = \phi(\nu_j^{(i)}) \). For concreteness, we take \( \phi = \tanh \), though the method applies to any monotonic sector bounded activation. Figure~\ref{fig:top_three_images} illustrates the operations explained in this subsection. We aim to linearly upper and lower bound \( \phi(\nu_j^{(i)}) \) over its input interval using slope parameters:
\begin{equation}
    l_j^{(i)} y \leq \phi(\nu_j^{(i)}) \leq u_j^{(i)} y.
\end{equation}
The coefficients \( l_j^{(i)} \) and \( u_j^{(i)} \) are chosen based on the signs of \( \underline{\nu}_j^{(i)} \) and \( \overline{\nu}_j^{(i)} \), as follows:
\begin{equation}\label{eq:casesalphabeta}\small
\begin{cases}
    \beta_j^{(i)} \underline{\nu}_j^{(i)} \leq \phi(\nu_j^{(i)}) \leq \alpha_j^{(i)} \overline{\nu}_j^{(i)} & \text{if } \underline{\nu}_j^{(i)} \geq 0,\ \overline{\nu}_j^{(i)} > 0, \\
    \alpha_j^{(i)} \underline{\nu}_j^{(i)} \leq \phi(\nu_j^{(i)}) \leq \beta_j^{(i)} \overline{\nu}_j^{(i)} & \text{if } \underline{\nu}_j^{(i)} < 0,\ \overline{\nu}_j^{(i)} \leq 0, \\
    -|\alpha_{1,j}^{(i)} \underline{\nu}_j^{(i)}| \leq \phi(\nu_j^{(i)}) \leq |\alpha_{2,j}^{(i)} \overline{\nu}_j^{(i)}| & \text{if } \underline{\nu}_j^{(i)} < 0,\ \overline{\nu}_j^{(i)} > 0.
\end{cases}
\end{equation}

Each of the above cases corresponds to Figures \ref{fig:posnu}, \ref{fig:negnu}, and \ref{fig:unstablenode} with the parameters $\alpha$ and $\beta$ defined in the captions. The first two cases are handled by directly replacing \( \underline{\nu}_j^{(i)} \) and \( \overline{\nu}_j^{(i)} \) with the sector expressions in \eqref{eq:firstoftanh}:
\begin{equation}\label{eq:replacements}\small
\underline{\nu}_j^{(i)} = L_{j,:}^{(i)} y \quad \overline{\nu}_j^{(i)} = U_{j,:}^{(i)} y,
\end{equation}
respectively, where $L_{j,:}$ denotes the $j$-th row of matrix $L$. For the third case (crossing zero), we apply the Cauchy–Schwarz inequality:
\begin{equation}\small
    -|\alpha_{1,j}^{(i)}|\, |\underline{\nu}_j^{(i)}| \leq \phi(\nu_j^{(i)}) \leq |\alpha_{2,j}^{(i)}|\, |\overline{\nu}_j^{(i)}|.
\end{equation}
Now, we replace $\underline{\nu}_j^{(i)}$ and ${\overline{\nu}}_j^{(i)}$ by $L_{j,:}^{(i)}\,y$ and $U_{j,:}^{(i)}\,y$ respectively, and since \( y \geq 0 \) for positive systems, we can move the vector \( y \) outside the absolute value, yielding:
\begin{equation}\label{eq:thirdcasealphabeta}\small
    -|\alpha_{1,j}^{(i)}|\, |L_{j,:}^{(i)}| y \leq \phi(\nu_j^{(i)}) \leq |\alpha_{2,j}^{(i)}|\, |U_{j,:}^{(i)}| y.
\end{equation}

We now combine equations \eqref{eq:casesalphabeta}, \ref{eq:replacements}, and \eqref{eq:thirdcasealphabeta} to form a unified equation. Let \( D^{(i)}_{\text{low}}, D^{(i)}_{\text{up}} \in \mathbb{R}^{n_i \times n_i} \) be diagonal matrices containing the slope parameters and define adjusted slope matrices \( \hat{L}^{(i)}, \hat{U}^{(i)} \in \mathbb{R}^{n_i \times n_0} \). The post-activation output is sector bounded as:
\begin{equation}\label{eq:DandL}\small
    D^{(i)}_{\text{low}}\, \hat{L}^{(i)} y \leq \phi(\nu^{(i)}) \leq D^{(i)}_{\text{up}}\, \hat{U}^{(i)} y,
\end{equation}
where:
\begin{align}\label{case:cases}
\small
\begin{cases}
    D^{(i)}_{\text{low},jj} = \beta_j^{(i)},\quad D^{(i)}_{\text{up},jj} = \alpha_j^{(i)}, \\
    \hat{L}_{j,:}^{(i)} = L_{j,:}^{(i)},\quad \hat{U}_{j,:}^{(i)} = U_{j,:}^{(i)}
    & \text{if } \underline{\nu}_j^{(i)} \geq 0,\ \overline{\nu}_j^{(i)} > 0, \\[1ex]
    D^{(i)}_{\text{low},jj} = \alpha_j^{(i)},\quad D^{(i)}_{\text{up},jj} = \beta_j^{(i)}, \\
    \hat{L}_{j,:}^{(i)} = L_{j,:}^{(i)},\quad \hat{U}_{j,:}^{(i)} = U_{j,:}^{(i)}
    & \text{if } \underline{\nu}_j^{(i)} < 0,\ \overline{\nu}_j^{(i)} \leq 0, \\[1ex]
    D^{(i)}_{\text{low},jj} = -|\alpha_{1,j}^{(i)}|,\quad D^{(i)}_{\text{up},jj} = |\alpha_{2,j}^{(i)}|, \\
    \hat{L}^{(i)}_{j,:} = |L^{(i)}_{j,:}|,\quad \hat{U}^{(i)}_{j,:} = |U^{(i)}_{j,:}|
    & \text{if } \underline{\nu}_j^{(i)} < 0,\ \overline{\nu}_j^{(i)} > 0.
\end{cases}
\end{align}

The construction in \eqref{eq:DandL} enables the propagation of sector bounds through the activation function.
In the next step, we propagate these sector bounds through weight matrices.

\subsection{Propagation of Sector Bounds through Weight Matrices}\label{subsec:lintrans}
At this step, the bounds are passed through the weight matrices to obtain updated sector bounds. From the last subsection, we can rewrite equation \eqref{eq:DandL} as below:
\begin{equation*}\label{eq:startWs}
    l^{(i)} y \leq \omega^{(i)} \leq u^{(i)} y, \quad
l^{(i)} = D^{(i)}_{\text{low}}\, \hat{L}^{(i)}, \quad  u^{(i)} = D^{(i)}_{\text{up}}\, \hat{U}^{(i)}.
\end{equation*}




To propagate this sector constraint through the linear map \( W^{(i+1)} \), we use the positive and negative decompositions of the weight matrix:
\begin{align}\small
    &(W^{(i+1)}_+)_{jk} = \max(W^{(i+1)}_{jk}, 0),\nonumber\\
    &(W^{(i+1)}_-)_{jk} = \min(W^{(i+1)}_{jk}, 0),\nonumber
\end{align}
so that \( W^{(i+1)} = W^{(i+1)}_+ + W^{(i+1)}_- \).

Applying this decomposition, the propagated bounds on \( \nu^{(i+1)} \) are given by:
\begin{align*}\begin{aligned}\small
    W^{(i+1)}_+ l^{(i)} y + W^{(i+1)}_- u^{(i)} y 
    &\leq \nu^{(i+1)} \\
    &\leq W^{(i+1)}_- l^{(i)} y + W^{(i+1)}_+ u^{(i)} y.
\end{aligned}\end{align*}

Thus, the pre-activation logits \( \nu^{(i+1)} \) are sector bounded with respect to \( y \) as:
\begin{equation}\label{eq:lastofWs}
    L^{(i+1)} y \leq \nu^{(i+1)} \leq U^{(i+1)} y,
\end{equation}
where the new slope matrices are defined by:
{\small\begin{align*}
    L^{(i+1)} &= W^{(i+1)}_+ l^{(i)} + W^{(i+1)}_- u^{(i)}, \\
    U^{(i+1)} &= W^{(i+1)}_- l^{(i)} + W^{(i+1)}_+ u^{(i)}.
\end{align*}}

This recursive formulation enables the forward propagation of sector bounds across the weight matrices of each layer.

The equation \eqref{eq:lastofWs} can be fed into equation \eqref{eq:firstoftanh} to close the loop of iteration for layers.

\subsection{Sector Bound for the Entire Neural Network}
We now combine the results of the previous two subsections to derive a sector bound for the entire FFNN. We explain how to go step by step to get the whole sector bound.
We start from the sector bound of the first layer given in equation \eqref{eq:bound1}, and propagate it forward iteratively using the results of Subsections \ref{subsubsec:tanhrelax} (activation relaxations) and ~\ref{subsec:lintrans} (weight matrix transformations). This process yields a final sector bound for the overall neural network output with respect to its input. We summarize the result in the following.


\begin{theorem}[Local Sector Bound for FFNN]\label{thm:nn_sector_bound}
Consider an FFNN as in \eqref{eq:NNcontroller} $\pi:\mathbb{R}_+^p\to\mathbb{R}^m$ with sector bounded activation functions. Given the set $\Gamma$, then the sector bound for the network mapping \( \mathrm{NN}(y) \) can be calculated as:
{\small\begin{align}\label{eq:localsectorbound}
&\gamma_1\, y \leq \mathrm{NN}(y) \leq \gamma_2\, y, \quad \forall y \in \Gamma,\\
&\gamma_1 = L^{(q+1)} \left( \prod_{i=2}^{q} D_{\text{low}}^{(i)} \hat{L}^{(i)} \right) W^{(1)},\nonumber\\
&\gamma_2 = U^{(q+1)} \left( \prod_{i=2}^{q} D_{\text{up}}^{(i)} \hat{U}^{(i)} \right) W^{(1)},\nonumber \end{align}}

and the matrices \( D_{\text{low}}^{(i)}, D_{\text{up}}^{(i)}, \hat{L}^{(i)}, \hat{U}^{(i)} \) are defined according to the activation relaxation scheme described in \eqref{case:cases}.
\end{theorem}

This local sector bound enables the application of the local positive Aizerman conjecture for analyzing the stability of NN feedback systems. In particular, we utilize Theorem~\ref{the:mathmanip} and Lemma~\ref{lem:mathmanip} to certify local exponential stability and estimate the ROA of NN-controlled feedback loops.

\begin{corollary}\label{cor:localpositiveaizerman}
Consider the Lur'e system defined in \eqref{eq:positiveluresystem}, where the nonlinearity \( \Phi \) is realized by an FFNN as described in \eqref{eq:NNcontroller}. Given a compact set \( \Gamma \) if the matrix \( A + B\gamma_1C \) is Metzler and \( A + B\gamma_2C \) is Hurwitz, with slope matrices \( \gamma_1, \gamma_2 \) computed via \eqref{eq:localsectorbound},
then the closed loop system is locally exponentially stable for all trajectories that remain within the set \( \Gamma \).
\end{corollary}

This result follows directly from Theorem~\ref{the:mathmanip} and Theorem~\ref{thm:nn_sector_bound}. Furthermore, in the scalar case (i.e., when the Lur’e system is MISO), Lemma~\ref{lem:mathmanip} can be applied to estimate the region of attraction. To do so, one computes the upper bound of \( \Gamma \), and then applies the analytical formula in \eqref{eq:ROA} to determine the set of stable initial conditions.

\vspace{-.1cm}
\section{Numerical Example}
Consider a Lur’e system as in \eqref{eq:positiveluresystem}, where an open-loop unstable LTI plant is stabilized by an NN controller. The plant dynamics are given by:  
\begin{equation*}\small
A = \begin{bmatrix}
    -7 & 5 \\
     6 & 1
\end{bmatrix}, \quad
B = \begin{bmatrix}
    1 \\
    2
\end{bmatrix}, \quad
C = \begin{bmatrix}
    1 & 1
\end{bmatrix}.
\end{equation*}
Using parametric analysis, we compute the minimal lower sector \( \Sigma_1 = -3 \) such that the matrix \( A + B\Sigma_1 C \) becomes Metzler, and the maximal upper sector \( \Sigma_2 = -1.276 \) such that \( A + B\Sigma_2 C \) is Hurwitz. These values define the stable sector interval \( [\Sigma_1, \Sigma_2] \) for the application of Aizerman-type results. A PID controller with gains \( P = -1.25 \), \( I = -0.1 \), and \( D = -0.1 \) is used as a benchmark for generating training data for a stabilizing NN controller. We collect closed loop input-output trajectory data by randomly initializing the system near \( x_0 = [2,\ 2]^\top \). An FFNN is trained to approximate the stabilizing PID controller. Figure~\ref{fig:NNandsigmas} shows the resulting NN and its input-output relation relative to the stable sector bounds \( [\Sigma_1, \Sigma_2] \). It is evident that the trained NN is  $\Gamma-$sector bounded within the interval. Numerically, we identify the $\Gamma$ set and its maximal element to be \( \overline{y} = 12.2 \).

\begin{figure}[h]
\renewcommand\thesubfigure{\thefigure\alph{subfigure}} 
\makeatletter
\renewcommand{\p@subfigure}{} 
\makeatother
    \centering
    \begin{subfigure}[t]{0.49\linewidth}
        \centering
        \includegraphics[width=\linewidth]{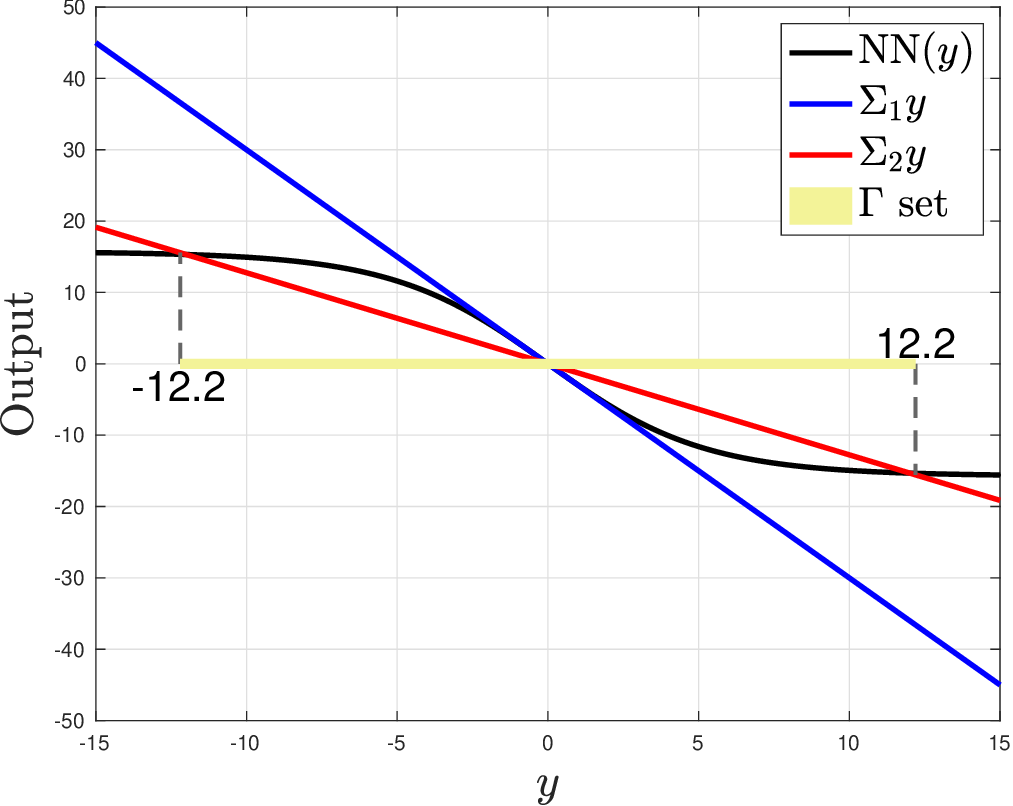}
        \caption{NN output, safe sector bounds, and corresponding $\Gamma$ set.\vspace{-.3cm}}
        \label{fig:NNandsigmas}
    \end{subfigure}
    \hfill
    \begin{subfigure}[t]{0.49\linewidth}
        \centering
        \includegraphics[width=\linewidth]{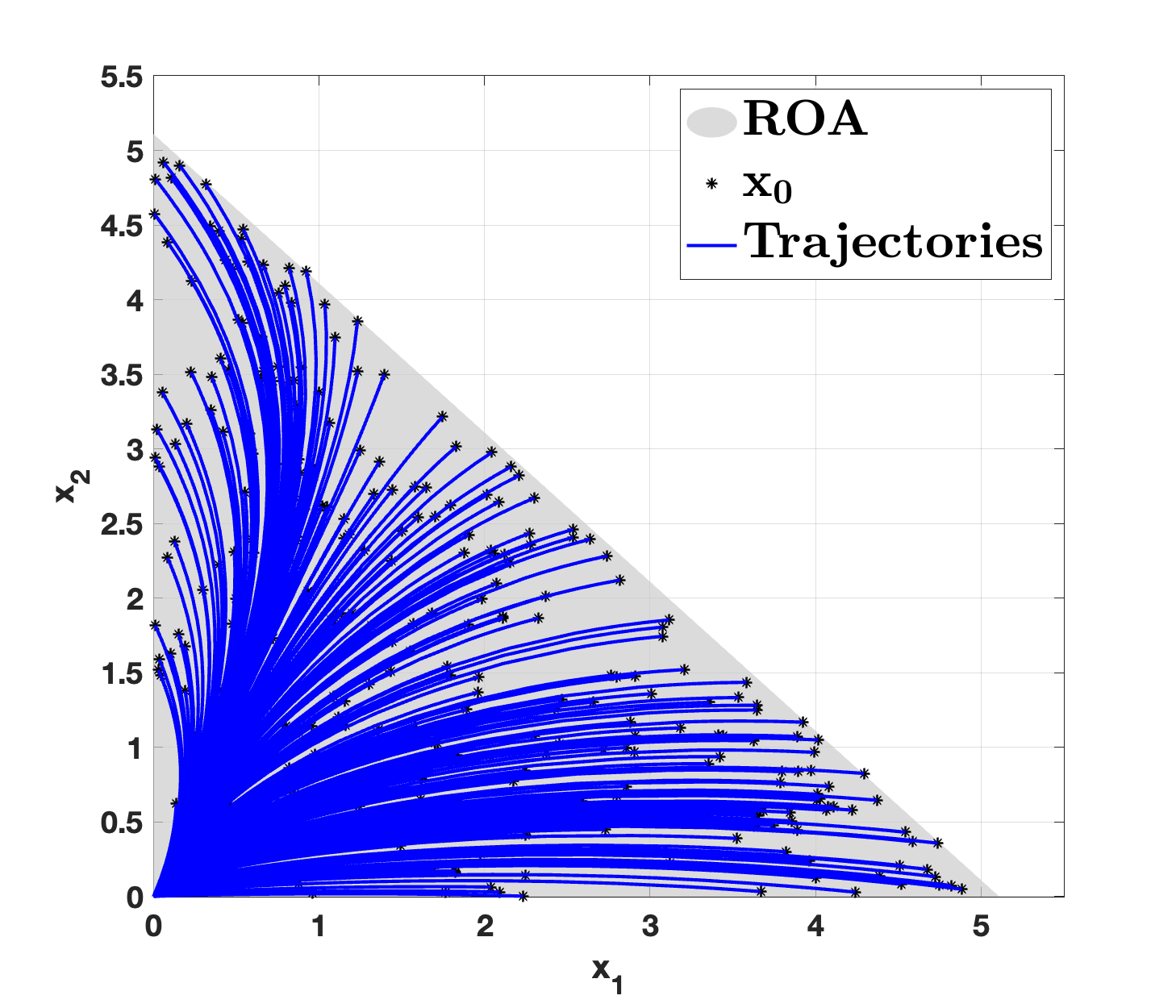}
        \caption{Stability of trajectories inside the ROA estimated by Lemma~\ref{lem:mathmanip}.\vspace{-.3cm}}
        \label{fig:mathmanipROA}
    \end{subfigure}
\end{figure}

\paragraph{Region of Attraction via Local Positive Aizerman}  
Regarding Lemma~\ref{lem:mathmanip}, we compute $v$ with the max ratio \( v_m / v_M = 0.42 \). Applying equation~\eqref{eq:ROA}, we obtain the estimate of ROA as
\(
Cx_0 \leq 5.12.
\)
Figure~\ref{fig:mathmanipROA} shows closed-loop trajectories demonstrating exponential convergence for initial conditions within this ROA.

\paragraph{Region of Attraction via Lyapunov Method}
We now apply the Lyapunov-based approach by solving the matrix inequality \eqref{eq:condit3}. This yields a quadratic Lyapunov function
\begin{equation*}\small
    V(x) = x^\top P_{\text{lm}} x, \textrm{ with } P_{\text{lm}} = \begin{bmatrix}
    0.1235 & 0.1323 \\
    0.1323 & 0.3919
\end{bmatrix}
\end{equation*}
Using Algorithm~\ref{alg:lyapunov_roa}, we estimate the largest sublevel set satisfying \( \dot{V}(x) < 0 \), resulting in the ROA
$x^\top P_{\text{lm}} x \leq 32.$
Figure~\ref{fig:lyapvdotregions} visualizes this approach by showing sublevel sets that remain in the negative $\dot V$ area, and Figure~\ref{fig:lyapsimul} shows random trajectories starting inside the estimated ROA.\vspace{-.1cm}
\begin{figure}[h]
\renewcommand\thesubfigure{\thefigure\alph{subfigure}} 
\makeatletter
\renewcommand{\p@subfigure}{} 
\makeatother
    \centering
    \begin{subfigure}[t]{0.45\linewidth}
        \centering
        \includegraphics[width=\linewidth]{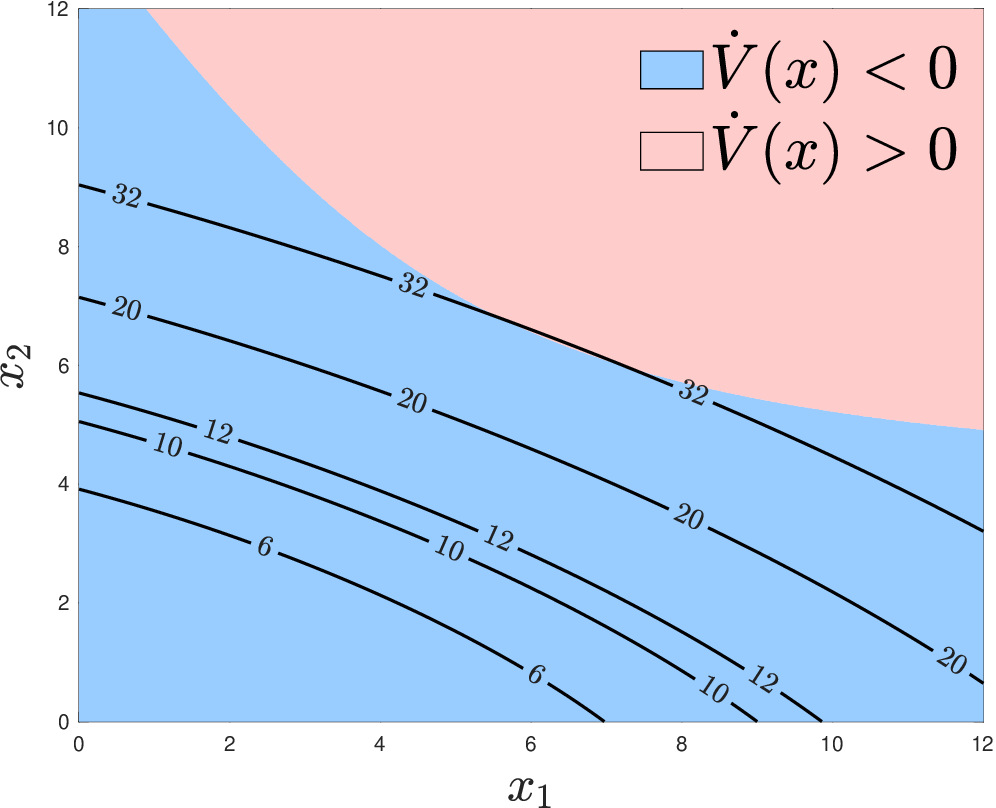}
        \caption{Evaluation of $\dot{V}$ from equation~\eqref{eq:vdotoflyap} across $\mathbb{R}^2$ with overlaid sublevel sets of $V$.\vspace{-.3cm}}
        \label{fig:lyapvdotregions}
    \end{subfigure}
    \hfill
    \begin{subfigure}[t]{0.45\linewidth}
        \centering
        \includegraphics[width=\linewidth]{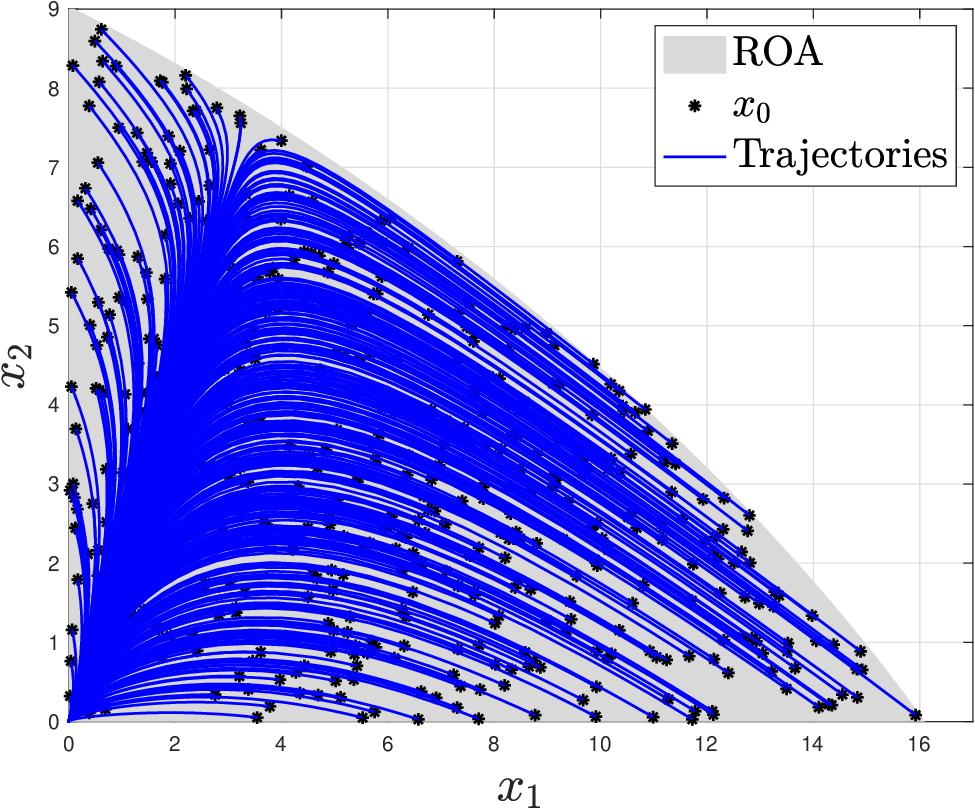}
        \caption{Stability of trajectories inside the ROA estimated by Algorithm \ref{alg:lyapunov_roa}.\vspace{-.3cm}}
        \label{fig:lyapsimul}
    \end{subfigure}
    \vspace{-.4cm}
\end{figure}

\begin{figure}[h]
    \centering
    \includegraphics[width=\linewidth]{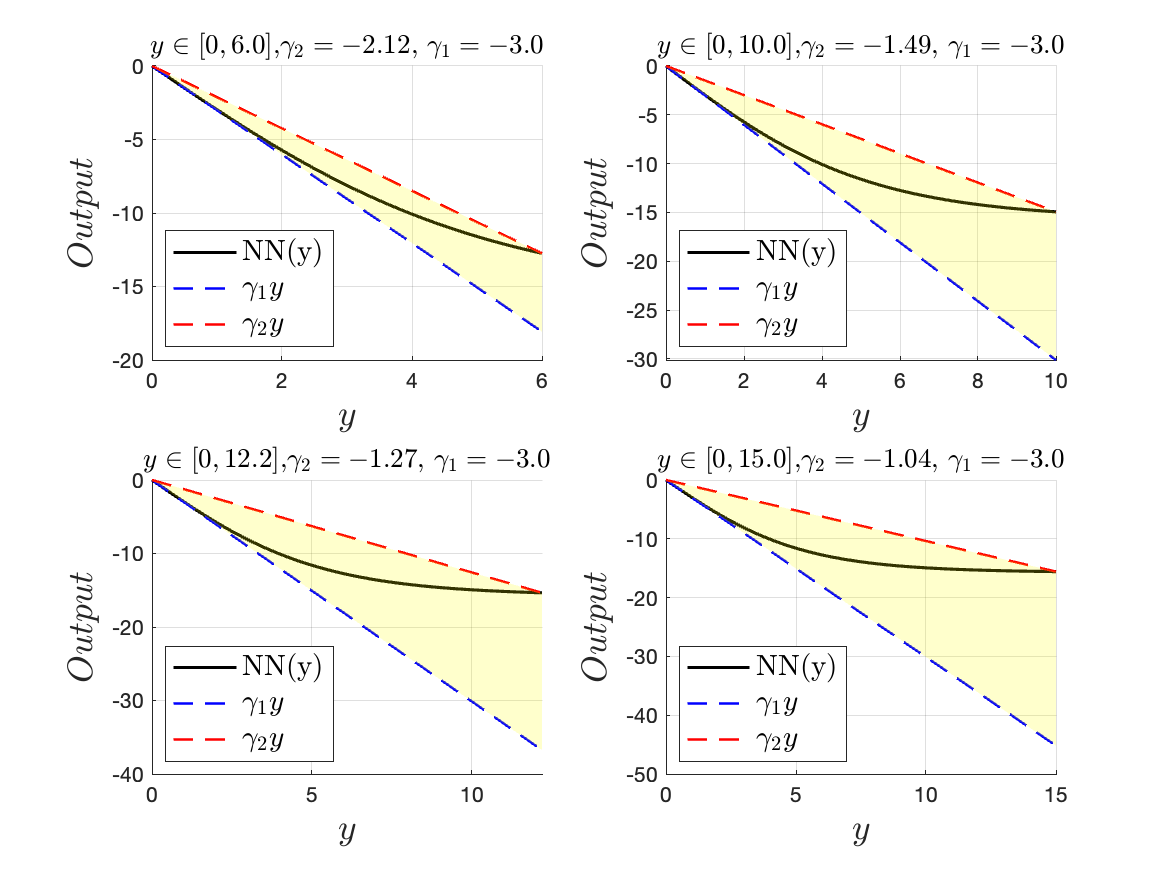}
    \caption{Local sector bounds calculated for given $y\in\Gamma$ using equation \eqref{eq:localsectorbound}. The title of each subplot defines the $\Gamma$ set and the corresponding $[\gamma_1,\gamma_2]$.\vspace{-.3cm}}
    \label{fig:localsectorbounds}
\end{figure}

\paragraph{Local Sector Bound}
We now evaluate our proposed local sector bound formulation. As shown in Figure~\ref{fig:localsectorbounds}, our local sector bounds tightly follow the NN output. This figure demonstrates how the sector bounds evolve as a function of the input range. As expected, the bounds widen with increasing input magnitude and intersect the critical upper sector \( \Sigma_2 = -1.276 \) at \( y\in [0,12.2] \). This value represents the largest admissible \( \Gamma \) set. Applying Corollary \ref{cor:localpositiveaizerman} and Lemma~\ref{lem:mathmanip} using \( \overline{y} = 12.2 \) recovers the same ROA plot as shown in Figure~\ref{fig:mathmanipROA}. Moreover, Figure~\ref{fig:bounds} compares our local sector bounds with the CROWN and IBP bounds calculated using the auto-LIRPA library \cite{wang2021beta}, showing structural differences.

\paragraph{Comparison with IQC-Based Approach}
As a benchmark, we compare our method against one of the best approaches in the literature \cite{yin2021stability}, which uses IQCs to bound NN nonlinearities. This method produces a Lyapunov function
\begin{equation*}\small
    V(x) = x^\top P_{\text{qc}} x, \textrm{ with } P_{\text{qc}} = \begin{bmatrix}
    0.1675  & -0.0224 \\
   -0.0224  &  0.0668
\end{bmatrix},
\end{equation*}
and yields a conservative ROA of
\(
x^\top P_{\text{qc}} x \leq 1.
\)

Table~\ref{tab:comparison} summarizes the comparison of methods in terms of runtime and estimated ROA, while Figure~\ref{fig:allROAs} visualizes the corresponding ROA estimates. As shown, the local sector bound approach outperforms the IQC-based method in both scalability and conservatism of ROA estimation. As expected, the Lyapunov method provides the largest underapproximation at the cost of computational complexity.
\vspace{-.5cm}
\begin{table}[h]
{\small
    \centering
    \begin{tabularx}{\linewidth}{|X|c|c|}
        \hline
        \centering \textbf{Method} & \textbf{Runtime (s)} & \textbf{ROA} \\ \hline
        Local Sector Bound & $5 \times 10^{-4}$ & $Cx_0 \leq 5.12$ \\
        Lyapunov-based Method & $5.45$ & $x^\top P_{\text{lm}} x \leq 32$ \\
        IQC-based & $1.03$ & $x^\top P_{\text{qc}} x \leq 1$ \\ \hline
    \end{tabularx}}
    \caption{Comparison of methods based on runtime and ROA.\vspace{-.4cm}}
    \label{tab:comparison}
\end{table}

\vspace{-.5cm}
\begin{figure}[h]
\renewcommand\thesubfigure{\thefigure\alph{subfigure}} 
\makeatletter
\renewcommand{\p@subfigure}{} 
\makeatother
    \centering
    \begin{subfigure}[t]{0.49\linewidth}
        \centering
        \includegraphics[width=\linewidth]{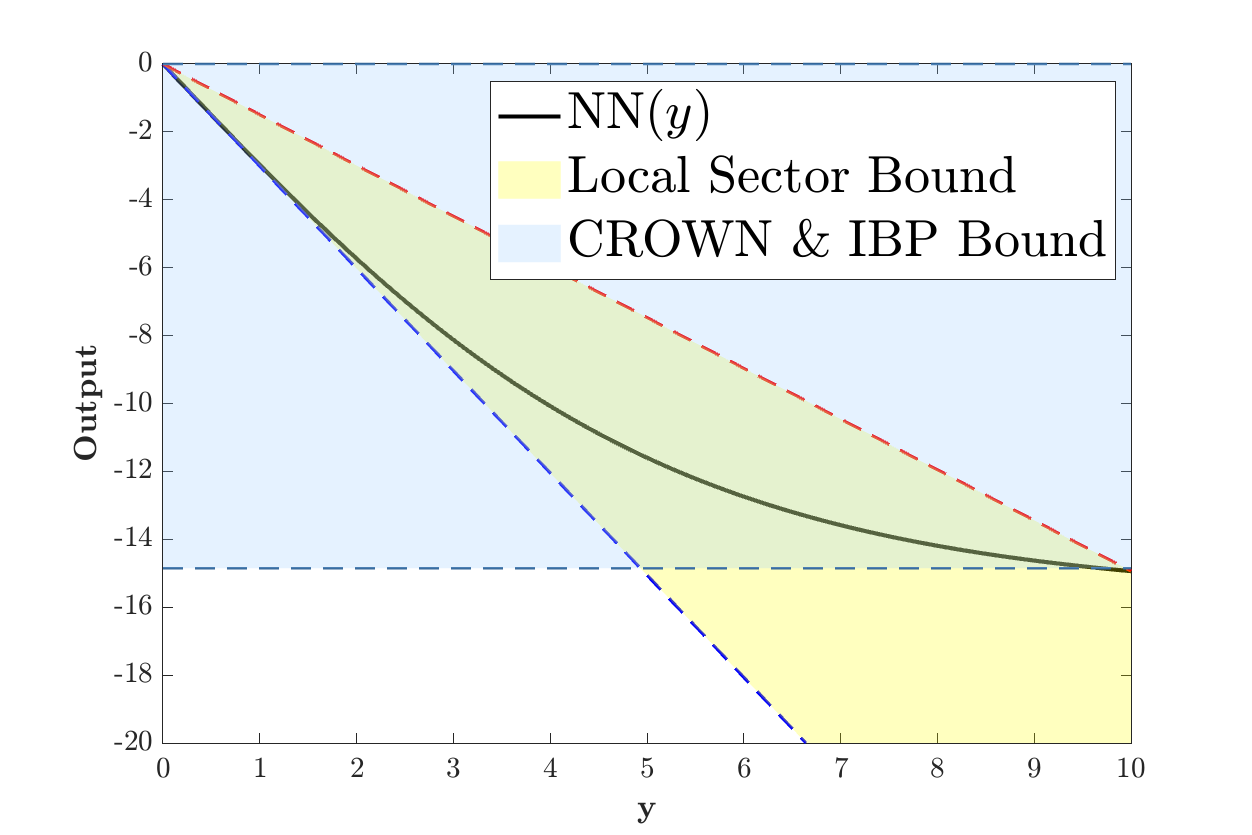}
        \caption{Comparison of local sector bounds vs. CROWN \& IBP bounds.\vspace{-.3cm}}
        \label{fig:bounds}
    \end{subfigure}
    \hfill
    \begin{subfigure}[t]{0.45\linewidth}
        \centering
        \includegraphics[width=\linewidth]{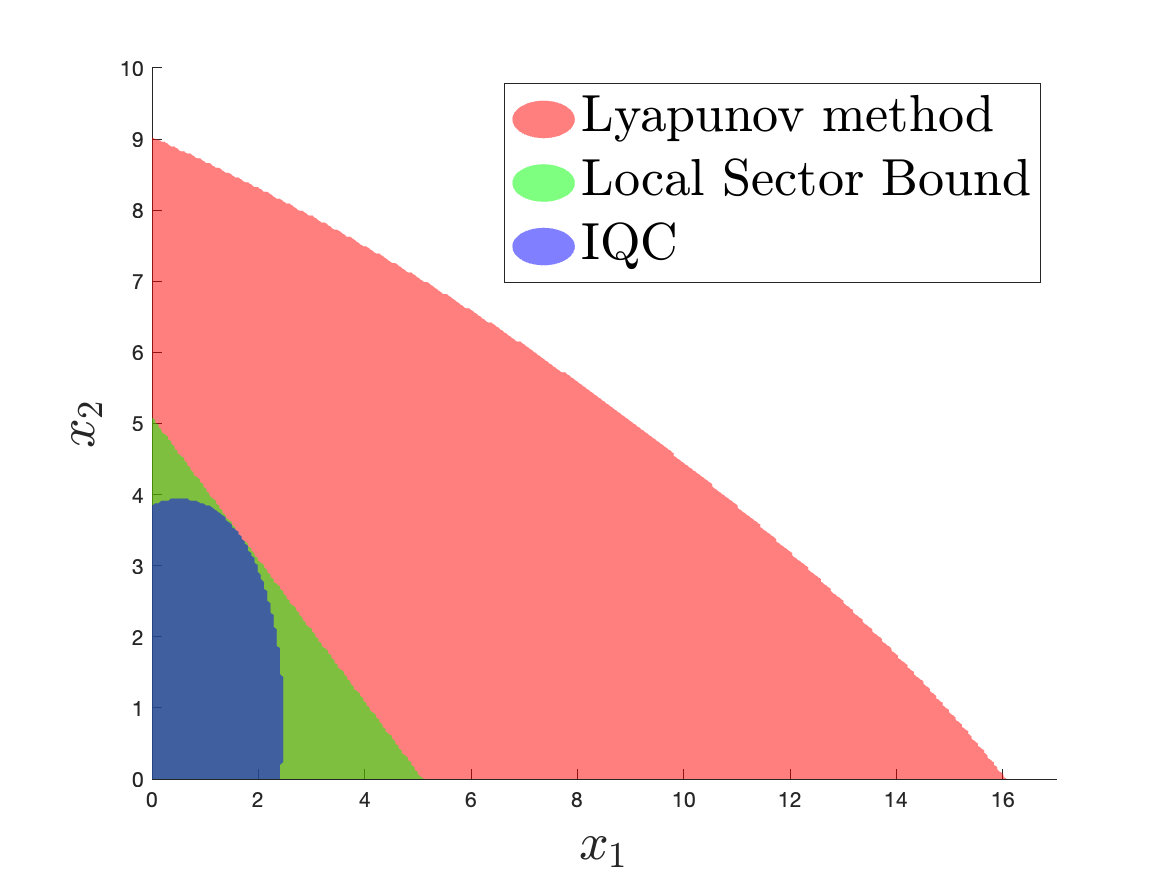}
        \caption{Comparison of our estimated ROAs vs. IQC method ROA.\vspace{-.3cm}}
        \label{fig:allROAs}
    \end{subfigure}
\end{figure}

\vspace{-.3cm}
\section{Conclusion and Future Work}
This work presents a comprehensive framework for analyzing the local stability of positive Lur’e systems with NN feedback. By introducing a localized version of the positive Aizerman conjecture, we derive verifiable conditions under which output trajectories confined to a compact set exhibit exponential stability. This theoretical foundation enables us to propose two distinct methods for estimating the ROA around a stable equilibrium.
The first method leverages a Lyapunov-based approach, where we construct a quadratic Lyapunov function via LMIs and identify invariant sublevel sets guaranteeing convergence. This method is applicable to arbitrary sector-bounded nonlinearities, including NNs, and benefits from strong guarantees under classical positive system theory.
The second method introduces a \emph{novel sector bound for FFNNs}. By propagating local input bounds through the network using layer-wise linear relaxations, we derive tight, structure-dependent sector bounds without relying on global Lipschitz constants or loose norm-based approximations. This bound is seamlessly integrated into the localized Aizerman framework, enabling efficient stability certification and ROA estimation that scales with network size and complexity.
Simulation results demonstrate that both proposed methods outperform existing IQC-based techniques in terms of ROA size and computational efficiency.

Future work includes extending the framework to biased and recurrent neural networks, incorporating the novel bounds in other stability verification frameworks, and employing this approach for training safe learning-based control policies.
\begin{spacing}{.8}
\bibliography{References} 
\end{spacing}

\end{document}